\documentclass[a4paper,9pt]{extarticle}
\usepackage{spconf,amsmath,amsfonts,amsthm,amssymb,graphicx}
\usepackage{algorithm,algorithmicx,algpseudocode}
\usepackage{times}
\usepackage[font={small}]{caption}
\usepackage{subfigure}
\usepackage{wrapfig}
\usepackage{caption}
\usepackage{transparent}
\usepackage{xcolor,colortbl}
\usepackage{dcolumn}

\usepackage[compact]{titlesec}
\titlespacing{\subsection}{0pt}{1ex}{0ex}
\titlespacing{\subsubsection}{0pt}{0.5ex}{0ex}
 
\title{A simple numerically stable primal-dual algorithm for computing
  Nash-equilibria in sequential games with incomplete information}

\name{Elvis Dohmatob}
\address{Parietal Team, INRIA / CEA, Universit\'e de Paris-Saclay}
\DeclareMathOperator{\proj}{proj}

\DeclareMathOperator{\prox}{prox}

\newtheorem{definition}{Definition}
\newtheorem{theorem}{Theorem}
\newtheorem{remark}{Remark}

\begin{document}

\maketitle
\begin{abstract}
We present a simple primal-dual algorithm for computing approximate
Nash-equilibria in two-person zero-sum sequential games with
incomplete information and perfect recall (like Texas Hold'em
Poker). Our algorithm is numerically stable, performs only basic
iterations (i.e matvec multiplications, clipping, etc., and no calls
to external first-order oracles, no matrix inversions, etc.), and is
applicable to a broad class of two-person zero-sum games including
simultaneous games and sequential games with incomplete information
and perfect recall. The applicability to the latter kind of games is
thanks to the sequence-form representation which allows us to encode
any such game as a matrix game with convex polytopial strategy
profiles. We prove that the number of iterations needed to produce a
Nash-equilibrium with a given precision is inversely proportional to
the precision. As proof-of-concept, we present experimental results on
matrix games on simplexes and Kuhn Poker.
\end{abstract}

\section{Introduction}
\label{sec:intro}
A game-theoretic approach to playing games strategically optimally
consists of computing Nash-equilibria (in fact, approximations thereof)
offline, and playing one's part (an optimal \textit{behavioral}
strategy) of the equilibrium online. This
technique is the driving-force behind solution concepts like CFR
\cite{zinkevich2008regret,lanctot2009monte,Bowling09012015},
$\text{CFR}^{+}$ \cite{tammelin14} and other variants, which
have recently had profound success in Poker. However, solving games
for equilibria remains a mathematical and computational challenge,
especially in sequential games with imperfect information. In this
paper, we propose (our detailed contributions are sketched in
subsection \ref{sec:sketch} below and elaborated in section
\ref{sec:gsp}) a simple primal-dual algorithm for solving for such
equilibria approximately (in a sense to be made precise in Definition
\ref{thm:cool_notion} below).

\subsection{Statement of the problem}
The sequence-form representation for two-person zero-sum games with
incomplete information was introduced in
\cite{koller1992complexity}, and the theory was further developed in
\cite{koller1994fast,von1996efficient,vonequilibrium} where it was
established that for such games, there exist sparse matrices
$A \in \mathbb{R}^{n_1 \times n_2}$, $E_1 \in \mathbb{R}^{l_1 \times
  n_1}$, $E_2 \in \mathbb{R}^{l_2 \times n_2}$, and vectors $e_1 \in
\mathbb{R}^{l_1}, e_2 \in \mathbb{R}^{l_2}$ such that $n_1$, $n_2$,
$l_1$, and $l_2$ are all linear in the size of the game tree (number
of states in the game) and such that Nash-equilibria correspond to
pairs $(x, y)$ of \textit{realization plans} which solve the primal
LCP (Linear Convex Program)
\begin{equation}
  \begin{aligned}
    \underset{(y,p) \in \mathbb{R}^{n_2} \times
     \mathbb{R}^{l_1}}{\text{minimize }}\langle e_1,
    p\rangle \hspace{.5em}\text{
       subject to: } &y \ge 0, E_2y = e_2,\\
    &-Ay + E_1^Tp \geq 0,
  \end{aligned}
  \label{eq:primal_pb}
\end{equation}

and the dual LCP
\begin{equation}
  \begin{aligned}
    \underset{(x,q) \in \mathbb{R}^{n_1} \times
      \mathbb{R}^{l_2}}{\text{maximize }}-\langle e_2,
    q\rangle\hspace{.5em}\text{subject
      to: } &x \ge
    0, E_1x = e_1,\\
    &A^Tx + E_2^Tq \geq 0.
  \end{aligned}
  \label{eq:dual_pb}
\end{equation}
The vectors $p = (p_0, p_1, ..., p_{l_2 - 1}) \in \mathbb{R}^{l_2}$
and $q = (q_0, q_1, ..., q_{l_1 - 1}) \in \mathbb{R}^{l_1}$ are dual
variables. 
$A$ is the \textit{payoff matrix} and each $E_k$ is a matrix whose
entries are $-1$, $0$ or $1$, with exactly 1 entry per row which
equals $-1$ except for the first whose whose first entry is $1$ and all
the others are $0$. Each of the vectors $e_k$ is of the form $(1, 0, ..., 0)$.

The LCPs above have the equivalent saddle-point formulation
\begin{equation}
  \underset{y \in Q_2}{\text{minimize}}\text{ }\underset{x \in
    Q_1}{\text{maximize}}\text{ }\langle x, Ay\rangle,
  \label{eq:gilpin}
\end{equation}
where the compact convex polytope
\begin{equation}
  Q_k := \{z \in \mathbb{R}^{n_k} | z \ge 0, E_kz = e_k\} \subseteq
  [0, 1]^{n_k}
\label{eq:polytope}
\end{equation}
is identified with the strategy profile of player $k$ in the
sequence-form representation. At a feasible point $(y, p, x, q)$ for
the LCPs, the \textit{duality gap} $\tilde{G}(y, p, x, q)$ is
given by\footnote{The first inequality being due to \textit{weak
    duality}.}
\begin{eqnarray}
  \begin{split}
  0 &\le \tilde{G}(y, p, x, q) := \langle e_1, p\rangle - (-\langle
  e_2, q\rangle) = \langle e_1, p\rangle + \langle
  e_2, q\rangle\\
  &= G(x, y) := \mathrm{max}\{\langle u, Ay\rangle - \langle x, Av\rangle |
(u,v) \in Q_1 \times Q_2\}.
\end{split}
  \label{eq:dgap}
\end{eqnarray}

In \eqref{eq:dgap}, the quantity $G(x, y)$ is nothing but the primal-dual
gap for the equivalent saddle-point problem \eqref{eq:gilpin}.
It was shown (see Theorem 3.14 of \cite{vonequilibrium}) that a pair
$(x, y) \in Q_1 \times Q_2$ of realization plans is a solution to the
LCPs \eqref{eq:primal_pb} and \eqref{eq:dual_pb} (i.e is a
Nash-equilibrium for the game)  if and only if there exist vectors $p$
and $q$ such that
\begin{equation}
\begin{split}
\hspace{.25em} -Ay + E_1^Tp \ge 0, &\hspace{.5em}A^Tx + E_2^Tq \ge
0, \hspace{.25em} \langle x, -Ay + E_1^Tp\rangle = 0,\\
&\langle y, A^Tx  + E_2^Tq\rangle = 0.
\end{split}
\label{eq:feasibility}
\end{equation}
Moreover, at equilibria \textit{strong duality} holds and the value
of the game equals $p_0 = -q_0$, i.e the duality gap
$\tilde{G}(y, p, x, q)$ defined in \eqref{eq:dgap} vanishes at
equilibria.

\begin{definition}[\textbf{Nash $\epsilon$-equilibria}]
Given $\epsilon > 0$, a Nash $\epsilon$-equilibrium is
a pair $(x^*, y^*)$ of realization plans for which there exists dual
vectors $p^*$ and $q^*$ for problems \eqref{eq:primal_pb} and
\eqref{eq:dual_pb} such that the duality gap at $(y^*, p^*, x^*, q^*)$
doesn't exceed $\epsilon$. That is,

\begin{equation}
  0 \le \tilde{G}(y^*, p^*, x^*, q^*) \le \epsilon.
\label{eq:approx_pb}
\end{equation}
\label{thm:approx_nash}
\end{definition}

\subsection{A remark concerning matrix games on simplexes}
It should be noted that any matrix $A \in \mathbb{R}^{n_1 \times n_2}$
specifies a matrix game with payoff matrix $A$, for which player $k$'s
strategy profile is a simplex
\begin{eqnarray}
\Delta_{n_k} := \left\{z \in \mathbb{R}^{n_k}|z \ge 0,\sum_j
z_j = 1\right\}.
\end{eqnarray}
This simplex can be written as a compact convex polytope in
the form \eqref{eq:polytope} by taking $E_k := (1, 1, ..., 1) \in
\mathbb{R}^{1 \times n_k}$ and $e_k = 1 \in \mathbb{R}^1$. Thus every
matrix game on simplexes can be seen as a sequential game, and so the
results presented in this manuscript can be trivially applied such
games in particular. For this special sub-class of sequential games,
the duality gap function $G(x,y)$ writes
\begin{eqnarray}
\begin{split}
G(x, y) &=
\mathrm{max}\{\langle u, Ay\rangle - \langle x, Av\rangle | (u,v) \in
\Delta_{n_1} \times \Delta_{n_2}\}\\
&= \underset{0 \le i <
  n_1}{\text{max }}(Ay)_i - \underset{0 \le j < n_2}{\text{min
}}(A^Tx)_j.
\end{split}
\label{eq:mg_pd}
\end{eqnarray}

\subsection{Quick sketch of our contribution}
\label{sec:sketch}
We now give a brief overview of our contributions, which will be made
more elaborate in section \ref{sec:gsp}.
Developing on an alternative notion of approximate equilibria (see
Definition \ref{thm:cool_notion})
homologous to that presented in Definition \ref{thm:approx_nash}, we
device a simple numerically stable primal-dual algorithm that
(Algorithm \ref{Tab:algo}) for computing approximate Nash-equilibria
in sequential two-person zero-sum games with incomplete information and
perfect recall. On, each iteration, the only operations performed by
our algorithm are of the form $A^Tx$, $Ay$, $E_1^Tp$, $E_2^Tq$, and
$(x)_+ := (\max(0, x_j))_j$. We also prove (Theorem \ref{thm:pd}) that
--in an ergodic / Ces\`ario sense-- the number of iterations required
by the algorithm to produce an approximation equilibrium to a
precision $\epsilon$ is $\mathcal{O}(1/\epsilon)$, with explicit
values for the constants involved.

\subsection{Notation and terminology}
\paragraph*{General.} Let $m$ and $n$ be positive integers.
The components of a vector $z \in \mathbb{R}^n$ will be
denoted $z_0$, $z_1$, ..., $z_{n-1}$ (indexing begins from $0$,
not $1$). $\mathbb{R}^{n}_+ := \{z \in \mathbb{R}^{n}\text{ }|\text{ }
z \geq 0\}$ is the nonnegative $n$th \textit{orthant}.
$\|z\|$ denotes the $2$-\textit{norm} of $z$ defined by $\|z\| :=
\sqrt{\langle z, z\rangle}$. 
Given a matrix $A \in \mathbb{R}^{m \times n}$, its \textit{spectral
  norm}, denoted $\|A\|$, is
 defined to be the largest \textit{singular value} of $A$, i.e the
 largest \textit{eigenvalue} of $A^TA$ (or equivalently, of $AA^T$).

\paragraph*{Convex analysis.} 
Given a subset $C \subseteq \mathbb{R}^n$, $i_C$ denotes the
\textit{indicator function} of $C$ defined by $i_C(x) = 0 \text{ if }
x \in C\text{ and }+\infty\text{ otherwise}$.
At times, we will write $i_{x \in C}$ for $i_C(x)$ (to ease notation,
etc.). For example, we will write $i_{z \ge 0}$ for
$i_{\mathbb{R}^n_+}(z)$, etc.  The \textit{orthogonal projector} onto
$C$, is the ``closest-point'' map $\proj_C: \mathbb{R}^n \rightarrow
C, x \mapsto \underset{z \in C}{\text{argmin }}\frac{1}{2}\|z-x\|^2$.
Let $f : \mathbb{R}^n \rightarrow (-\infty, +\infty]$ be a
  convex function. The \textit{effective domain} of $f$, denoted
  $dom(f)$, is defined as
$dom(f) := \{x \in \mathbb{R}^n | f(x) < +\infty\}$.
 If $dom(f) \ne \emptyset$ then we say $f$ is \textit{proper}.
The \textit{subdifferential} of $f$ at a point $x \in \mathbb{R}^n$ is
defined by
$\partial f(x) := \{v \in \mathbb{R}^n | f(z)  \ge f(x) + \langle v, z -
x\rangle, \forall z \in \mathbb{R}^n\}$.
If $f$ is convex, its
\textit{proximal operator} is the function $\prox_f: \mathbb{R}^n
\rightarrow \mathbb{R}^n$ defined by $\prox_f(x): = \underset{z \in
  \mathbb{R}^n}{\text{argmin }}\frac{1}{2}\|z
  - x\|^2 + f(z)$.

We recommend \cite{rockafellar1997convex,combettes2011proximal} for a
more detailed exposition on convex analysis and its use in
modern optimization theory and practice.

\section{Prior work}
\label{sec:related_work}
Here, we present a selection of algorithms that is representative of the
efforts that have been made in the literature to compute Nash
$\epsilon$-equilibria for two-person zero-sum games with incomplete
information like Texas Hold'em Poker, etc.
It should be noted that the class of games considered
here (sequential games with incomplete information), the LCPs
\eqref{eq:primal_pb} and \eqref{eq:dual_pb} are exceedingly larger
than what state-of-the-art LCP and interior-point solvers can
handle (see \cite{hoda2010smoothing,gilpinfirst}).

\subsection{Regret minimization}
CFR (CounterFactual Regret minimization) \cite{zinkevich2008regret},
Monte Carlo CFR \cite{lanctot2009monte}, and CFR+
\cite{Bowling09012015}, by their large popularity, have become the
definitive state-of-the-art, and are particularly useful in
many-player games, since convex-analytical methods cannot help much in
such games. Also, they can be shown to converge to a Nash-equilibrium
provided each player uses a CFR scheme to play the game
\cite{zinkevich2008regret}, but have a much weaker
convergence theory. For example, \cite{lanctot2009monte} showed that
such schemes have a prohibitive running time of
$\mathcal{O}(1/\epsilon^2)$ to produce a Nash $\epsilon$-equilibrium.

\subsection{First-order methods}
In \cite{hoda2010smoothing}, a nested iterative procedure using the
Excessive Gap Technique (EGT) \cite{nesterov2005excessive} (EGT and
precursors are well-known to the signal-processing community
\cite{NESTA}) was used
to solve the equilibrium problem \eqref{eq:gilpin}.
The authors reported a $\mathcal{O}(1/\epsilon)$ convergence rate
(which derives from the general EGT theory) for the outer-most
iteration loop.
\cite{gilpinfirst} proposed a modified version of the techniques in
\cite{hoda2010smoothing} and  proved a $\mathcal{O}\left(\left(\|A\| /
\delta\right) \log \left(1 / \epsilon\right)\right)$ convergence rate in
terms of the number of calls made to a first-order oracle. Here
$\delta = \delta(A, E_1, E_2, e_1, e_2) > 0$ is a certain
\textit{condition number} for the game. The crux of their technique was to
observe that \eqref{eq:gilpin} can further be written a the minimization of
the duality gap function $G(x, y)$ (defined in \eqref{eq:dgap})
for the game\footnote{The minimizers of $G$ are precisely the
  equilibria of the game.}, viz
\begin{eqnarray}
\mathrm{minimize}\{G(x,y)|(x,y) \in Q_1 \times Q_2\},
\end{eqnarray}
and then show there exists a scalar
$\delta > 0$ such that for any pair of realization plans $(x, y) \in Q_1 \times Q_2$,
\begin{eqnarray}
\text{``distance between }(x, y)\text{ and set of
equilibria'' } \le G(x, y)/\delta.
\end{eqnarray}
Their
algorithm is then derived by iteratively applying Nesterov smoothing
\cite{nesterov2005a}
with a geometrically decreasing sequence of tolerance levels
$\epsilon_{n+1} = \epsilon_n / \gamma$ (with $\gamma > 1$)
$G$. However, there are a number of issues, most notably: \textit{(a)}
The constant $\delta > 0$ can be arbitrarily small, and so the factor
$\|A\| / \delta$ in the $\mathcal{O}\left(\left(\|A\| / \delta\right)
\log\left(1 / \epsilon\right)\right)$ convergence rate can be
arbitrarily large for ill-conditioned games.
\textit{(b)} The reported linear convergence rate is not in terms of
  basic operations (addition, multiplication, matvec, clipping, etc.),
  but in terms of the number of calls to a first-order oracle. Most
  notably, projections onto the polytopes $Q_k$ are computed on
  each iteration, a very hard sub-problem. 

Recently, \cite{kroer2015} proposed accelerations to first-order
methods for computing Nash-equilibria (including those just
discussed), by an appropriate choice of the underlying \textit{Bregman
distance} and the \textit{distance generating function} (essential
ingredients in EGT-type algorithms). These modifications provably
gain a constant factor in the worst-case convergence rate over the
original algorithm. 

\subsection{Primal-dual algorithms}
The primal-dual algorithm first developed in \cite{chambolle2010}, was
proposed in \cite{chambolle2014ergodic} as a way of solving matrix
games on simplexes. Notably, such matrix games on
simplexes are considerably simpler than the games considered
here. Indeed, the authors in \cite{chambolle2014ergodic} used the fact
that computing the orthogonal projection of a point onto a simplex can
be done in linear time as in \cite{duchi2008efficient}. In contrast,
no such efficient algorithm is known nor is likely to exist, for the
polytopes $Q_k$ defined in \eqref{eq:polytope}. 
That notwithstanding, such projections can still be done iteratively
using for example, the algorithm in proposition 4.2 of
\cite{combettes2010dualization} or the algorithms developed in
\cite{tran2015splitting}. Unfortunately, as with any nested iterative
scheme, one would have to solve this sub-problem with finer and finer
precision, rendering the overall solver impractical. One can also cite
\cite{nurminski2008}, in which the authors endeavored an iterative
projection algorithm onto polytopes in outer representation.

Other than the difficult projection sub-problem just discussed,
the duality gap might explode even at points arbitrarily
 close to the set of feasible points, leaving the algorithm with no
 indication whatsoever, on whether progress is being made.


\section{Our contributions}
\label{sec:gsp}
\subsection{Generalized Saddle-point Problem (GSP) equivalent for
  Nash-equilibrium LCPs}
In the next theorem, we show that the LCPs \eqref{eq:primal_pb}
and \eqref{eq:dual_pb} can be
conveniently written as a GSP in the sense of
\cite{he2013accelerating}. The crux of idea is to remove the linear
constraints in the definitions of the strategy polytopes $Q_k$, by
augmenting the payoff matrix to yield an equivalent saddle-point
problem. The result is an equivalent game with unbounded strategy
profiles (nonnegative orthants) with much simpler geometry.
We elaborate the construction in the following theorem.
\begin{theorem}
Define two proper closed convex functions
  \begin{eqnarray}
    \left.
    \begin{aligned}
      g_1: \mathbb{R}^{n_2} &\times \mathbb{R}^{l_1} \rightarrow
      (-\infty, +\infty], \hspace{1em} g_1(y, p) :=
        i_{y \ge 0} + \langle e_1,p\rangle\\
        g_2: \mathbb{R}^{n_1} &\times \mathbb{R}^{l_2} \rightarrow
        (-\infty, +\infty],\hspace{1em} g_2(x, q) :=
          i_{x \ge 0} + \langle e_2, q\rangle
    \end{aligned}
    \right\}
    \label{eq:things}
  \end{eqnarray}
Also define two bilinear forms $\Psi_1$, $\Psi_2: \mathbb{R}^{n_2}
\times \mathbb{R}^{l_1} \times \mathbb{R}^{n_1} \times
\mathbb{R}^{l_2} \rightarrow \mathbb{R}$ by letting
\begin{equation}
  \begin{split}  
    K :=
    \left[
      \begin{array}{cc}
        A & -E_1^T \\
        E_2 & 0
      \end{array}
      \right],\hspace{.2em}
    \Psi_1(y, p, x, q)
    := \left\langle \begin{bmatrix}x\\q\end{bmatrix},
      K\begin{bmatrix}y\\p\end{bmatrix}\right\rangle, 
\end{split}
\end{equation}
with $\Psi_2 = -\Psi_1$, and define the functions $\hat{\Psi}_1$,
$\hat{\Psi}_2
\mathbb{R}^{n_2} \times \mathbb{R}^{l_1} \times \mathbb{R}^{n_1}
\times \mathbb{R}^{l_2} \rightarrow (-\infty, +\infty]$ by
\begin{eqnarray}
  \begin{aligned}
    \hat{\Psi}_1(y, p, x, q) :=\begin{cases}
    \Psi_1(y, p, x, q)+ g_1(y, p), &\mbox{ if }y \ge 0,\\
    +\infty, &\mbox{ otherwise}\end{cases}\\
    \hat{\Psi}_2(y, p, x, q) :=\begin{cases}
    \Psi_2(y, p, x, q)+ g_2(x, q), &\mbox{ if }x \ge 0,\\
    +\infty, &\mbox{ otherwise.}\end{cases}
  \end{aligned}
\end{eqnarray}
Finally, define the sets $S_1 := \mathbb{R}^{n_2}_+ \times
\mathbb{R}^{l_1}$ and $S_2 := \mathbb{R}^{n_1}_+ \times
\mathbb{R}^{l_2}$, and consider the GSP($\Psi_1$, $\Psi_2$, $g_1$,
$g_2$): Find a quadruplet $(y^*,p^*, x^*, q^*) \in S_1 \times S_2$
s.t $\forall (y,p, x, q) \in S_1 \times S_2$, we have
\begin{eqnarray}
  \begin{split}
    &\hat{\Psi}_1(y^*, p^*, x^*, q^*) \le \hat{\Psi}_1(y, p, x^*,
    q^*),\hspace{1em}and\\
    &\hat{\Psi}_2(y^*, p^*, x^*, q^*)
    \le \hat{\Psi}_2(y^*, p^*, x, q).
  \label{eq:unconstrained_pb}
\end{split}
\end{eqnarray}
\label{thm:pd}
Then GSP($\Psi_1$,
  $\Psi_2$, $g_1$, $g_2$) is equivalent to the LCPs
  \eqref{eq:primal_pb} and \eqref{eq:dual_pb}, i.e
a quadruplet $(y^*,p^*, x^*, q^*) \in \mathbb{R}^{n_2}
  \times \mathbb{R}^{l_1} \times \mathbb{R}^{n_1} \times
  \mathbb{R}^{l_2}$ solves the LCPs
  \eqref{eq:primal_pb} and \eqref{eq:dual_pb} iff it solves
  GSP($\Psi_1$, $\Psi_2$, $g_1$, $g_2$). 
  \label{thm:pd}
\end{theorem}

\begin{proof}
It suffices to show that at any point $(y, p, x,
q) \in S_1 \times S_2$, the duality gap between the primal
LCP \eqref{eq:primal_pb} and the dual LCP \eqref{eq:dual_pb} equals
the duality gap of GSP($\Psi_1$, $\Psi_2$, $g_1$, $g_2$).
Indeed, the unconstrained objective in \eqref{eq:primal_pb}, say
$a(x,y)$, can be computed as
\begin{eqnarray*}
  \begin{aligned}
    &a(y,p) = \langle e_1,p\rangle + i_{y\ge 0} + i_{-Ay + E_1^Tp \ge 0} +
    i_{E_2y = e_2}\\
    &= g_1(y,p) + \underset{x' \geq
      0}{\text{max}}\text{ }\langle x',Ay - E_1^Tp\rangle +
    \underset{q'}{\text{max}}\text{ }\langle q',E_2y - e_2\rangle\\
    &= g_1(y,p) + \underset{x',
      q'}{\text{max}}\text{ }\langle x',Ay\rangle - \langle x',
    E_1^Tp\rangle + \langle q',E_2y\rangle\\
    &\hspace{10em}-
    (i_{x' \ge 0} + \langle e_2,q\rangle)\\
    &= g_1(y,p)
      - \underset{x',q'}{\text{min}}\text{ }\Psi_2(y, p, x', q') + g_2(x',
      q')\\
      &= g_1(y,p)
      - \underbrace{\underset{x',q'}{\text{min}}\text{
        }\hat{\Psi}_2(y, p, x', q')}_{\phi_2(y,p)}
      = g_1(y, p) - \phi_2(y, p).
  \end{aligned}
  \label{eq:a}
\end{eqnarray*}
Similarly, the unconstrained objective, say
$b(x, q)$, in the dual LCP \eqref{eq:dual_pb} writes
\begin{eqnarray*}
  \begin{aligned}
&b(x, q) = 
-\langle q, e_2\rangle -i_{x \ge 0} - i_{A^Tx+E_2^Tq \ge 0} -
 i_{E_1x = e_1}\\
 &= -g_2(x, q) + \underset{y' \geq
   0}{\text{min}}\text{ }\langle y', A^Tx + E_2^Tq\rangle +
 \underset{p'}{\text{min}}\text{ }\langle p', e_1-E_1x\rangle\\
    &= -g_2(x, q)
 +\underset{y',p'}{\text{min}}\text{ }\Psi_1(y', p', x, q) +
 g_1(y', p')\\
& = -g_2(x, q) +
 \underbrace{\underset{y',p'}{\text{min}}\text{ }\hat{\Psi}_1(y', p',
   x, q)}_{\phi_1(x, q)} = -g_2(x, q) + \phi_1(x, q). 
   \end{aligned}
\end{eqnarray*}
Thus, noting that $-\infty < \phi_1(x, q), \phi_2(y, p) < +\infty$
(so that all the operations below are valid),
one computes the duality gap between the primal LCP
\eqref{eq:primal_pb} and dual the LCP \eqref{eq:dual_pb} at $(y, p, x, q)$ as
\begin{eqnarray*}
  \begin{split}
    &a(y, p) - b(x, q) = g_1(y, p) - \phi_2(y, p) + g_2(x, q) - \phi_1(x,
  q) \\
  &= \Psi_1(y, p, x, q) +  g_1(y, p) - \phi_2(y, p) + \Psi_2(y, p, x,
  q) + g_2(x, q) \\
  &\hspace{2em}- \phi_1(x, q)\\
  &= \hat{\Psi}_1(y, p, x, q) - \phi_1(x,
  q) + \hat{\Psi}_2(y, p, x, q) - \phi_2(y, p)\\
  &= \text{duality gap of GSP}(\Psi_1, \Psi_2,
  g_1, g_2) \text{ at }(y, p, x, q),
  \end{split}
\end{eqnarray*}
where the second equality follows from 
$\Psi_1 + \Psi_2 := 0$.
\end{proof}

By Theorem \ref{thm:pd}, solving for a Nash-equilibrium for
the game is equivalent to solving the GSP
\eqref{eq:unconstrained_pb}, which as it turns out, is simpler
conceptually (e.g, we no longer need to compute the
complicated orthogonal projections $\proj_{Q_k}$). The rest of the
paper will be devoted to developing an algorithm for solving the latter.


\subsection{The proposed algorithm}
We now derive the algorithm (Algorithm \ref{Tab:algo}) for computing
Nash $(\epsilon, 0)$-equilibria and establish its theoretical
properties. The algorithm, which emerges as a synthesis of Theorem
\ref{thm:pd} above and ideas from \cite{he2013accelerating},  is
numerically stable and performs only basic iterations (i.e matvec
multiplications, clipping, etc., and no calls to external first-order
oracles, no matrix inversions, etc.).
\label{sec:algo}
\begin{definition}
  Given $\epsilon > 0$ and a function $f:\mathbb{R}^n
  \rightarrow (-\infty,+\infty]$, the $\epsilon$-enlarged
  subdifferential (or
  $\epsilon$-subdifferential, for short) of $f$ is the set-valued
  function defined by
  \begin{eqnarray}
\partial_\epsilon f(x):= \{v \in \mathbb{R}^n | f(z)
\ge f(x) + \langle v, z - x\rangle - \epsilon,\forall z \in
\mathbb{R}^n\}.
\end{eqnarray}
\end{definition}
The idea behind $\epsilon$-subdifferentials is the following. Say we wish
to minimize a convex function $f$. Replace the usual
necessary and sufficient condition ``$0 \in \partial f(x)$'' for the
optimality of $x$ with the weaker condition ``$\partial_\epsilon f(x)$
contains a sufficiently small vector $v$''. 
This approximation concept for subdifferentials yields yet another notion
of approximate Nash-equilibrium.
the following concept of approximate Nash-equilibria (refer to
\cite{he2013accelerating}), namely

\begin{definition}[\textbf{Nash $(\epsilon_1,\epsilon_2)$-equilibria}]
Given tolerance levels $\epsilon_1, \epsilon_2 > 0$, a Nash
$(\epsilon_1,\epsilon_2)$-equilibrium for the GSP \eqref{eq:unconstrained_pb}
is any quadruplet $(x^*, y^*, x^*, q^*)$ for which
there exists a perturbation vector $v^*$ such that
$\|v^*\| \le \epsilon_1$ and $v^* \in
\partial_{\epsilon_2}[\hat{\Psi}_1(., ., x^*, q^*) +
  \hat{\Psi}_2(y^*, p^*, ., .)](y^*,p^*,x^*,q^*)$.
Such a vector $v^*$ is called a Nash $(\epsilon_1,
\epsilon_2)$-residual at the point $(x^*,
y^*, x^*, q^*)$.
\label{thm:cool_notion}
\end{definition}

The above definition is a generalization of the notion of
Nash-equilibria since: \textit{(a)} exact Nash-equilibria correspond
to Nash $(0,0)$-equilibria, and \textit{(b)} Nash
$\epsilon$-equilibria (in the sense of Definition
\ref{thm:approx_nash}) correspond to Nash $(0,\epsilon)$-equilibria.

\begin{algorithm}
\caption{Primal-dual algorithm for computing Nash $(\epsilon,
  0)$-equilibria in two-person zero-sum sequential games}
\label{Tab:algo}
\begin{algorithmic}[1]
\Require $\epsilon > 0$; $(y^{(0)},p^{(0)},x^{(0)},q^{(0)}) \in \mathbb{R}^{n_2}
  \times \mathbb{R}^{l_1} \times \mathbb{R}^{n_1} \times
  \mathbb{R}^{l_2}$.
\Ensure A Nash $(\epsilon,0)$-equilibrium
$({y^*},{p^*},{x^*},{q^*}) \in S_1 \times S_2$ for
the GSP \eqref{eq:unconstrained_pb}.
\State  \textbf{Initialize:} $\lambda \leftarrow 1/\|K\|$, ${v}^{(0)}
\leftarrow 0$, $k \leftarrow 0$
\While{ $ k = 0$ or  $\frac{1}{k\lambda}\|v^{(k)}\| \ge \epsilon$}
\State $y^{(k + 1)} \leftarrow (y^{(k)} - \lambda (A^Tx^{(k)} +
E_2^Tq^{(k)}))_+$, \hspace{.5em}$p^{(k+1)} \leftarrow p^{(k)} -
\lambda(e_1-E_1x^{(k)})$
\State $x^{(k + 1)} \leftarrow (x^{(k)} + \lambda (Ay^{(k+1)} -
E_1^Tp^{(k+1)}))_+$, \hspace{.5em}$\Delta x^{(k+1)} \leftarrow
x^{(k+1)}-x^{(k)}$
\State $\Delta q^{(k+1)} \leftarrow \lambda (E_2y -
e_2)$, \hspace{.5em}$q^{(k+1)} \leftarrow q^{(k)} + \Delta q^{(k+1)}$
\State $y^{(k+1)} \leftarrow y^{(k+1)} - \lambda (A^T\Delta x^{(k+1)}
+ E_2^T\Delta q^{(k+1)})$, \hspace{.5em}$\Delta y^{(k+1)} \leftarrow
y^{(k+1)}-y^{(k)}$
\State $p^{(k+1)} \leftarrow p^{(k+1)} + \lambda E_1\Delta x^{(k+1)}$,
\hspace{.5em} $\Delta p^{(k+1)} \leftarrow p^{(k+1)}-p^{(k)}$
\State ${v}^{(k+1)} \leftarrow {v}^{(k)} + (\Delta
y^{(k+1)},\Delta p^{(k+1)},\Delta x^{(k+1)},\Delta q^{(k+1)})$ 
\State $k \leftarrow k + 1$
\EndWhile
\end{algorithmic}
\end{algorithm}

\begin{theorem}[Ergodic / Ces\`ario $\mathcal{O}(1/\epsilon)$ convergence]
Let $d_0$ be the euclidean distance between the starting point
$(y^{(0)},p^{(0)},x^{(0)},q^{(0)})$ of Algorithm \ref{Tab:algo} and the
set of equilibria for the GSP \eqref{eq:unconstrained_pb}.
Then given any $\epsilon > 0$, there exists an index
$k_0 \le \frac{2d_0\|K\|}{\epsilon}$ such that after $k_0$ iterations
the algorithm produces a quadruplet
$(y^{k_0},p^{k_0},x^{k_0},q^{k_0})$ and a vector $v^{k_0}$ such that
$\|v_a^{k_0}\| \le \epsilon$ and $v_a^{k_0} \in
\partial[\hat{\Psi}_1(., ., x^{k_0}, q^{k_0}) +
  \hat{\Psi}_2(y^{k_0}, p^{k_0}, ., .)](y^{k_0},p^{k_0},x^{k_0},q^{k_0})$,
where
\begin{eqnarray}
\label{eq:v_a}
v_a^{(k_0)} := \frac{1}{k\lambda}v^{(k_0)}.
\end{eqnarray}
Thus Algorithm \ref{Tab:algo} outputs a Nash
$(\epsilon,0)$-equilibrium for the GSP \eqref{eq:unconstrained_pb}
in at most $\frac{2d_0\|K\|}{\epsilon}$ iterations.
\end{theorem}

\begin{proof}
It is clear that the quadruplet $(\Psi_1, \Psi_2, g_1, g_2)$
satisfies assumptions B.1, B.2, B.3, B.5, and B.6 of
\cite{he2013accelerating} with $L_{xx} = L_{yy} = 0$ and $L_{xy} =
L_{yx} = \|K\|$. Now, one easily computes the proximal operator of
$g_j$ in closed-form as $\prox_{\lambda g_j}(a, b) \equiv ((a)_+,
b - \lambda e_j)$. With all these ingredients in place, Algorithm
\ref{Tab:algo} is then obtained from \cite[Algorithm
  T-BD]{he2013accelerating} applied on the GSP
\eqref{eq:unconstrained_pb} with the choice of parameters: $\sigma = 1
\in (0, 1]$, $\sigma_x = \sigma_y = 0 \in [0, \sigma)$,
    $\lambda_{xy} := \frac{1}{\sigma L_{xy}}\sqrt{(\sigma^2 -
        \sigma_x^2)(\sigma^2 - \sigma_y^2)} = \sigma / \|K\| =
      1/\|K\|$, and $\lambda = \lambda_{xy} \in (0,
      \lambda_{xy}]$. The convergence result then follows immediately
  from \cite[Theorem 4.2]{he2013accelerating}.
\end{proof}

\subsection{Practical considerations}
\paragraph*{Efficient computation of $Ay$ and $A^Tx$.}
In Algorithms \ref{Tab:algo}, most of the time is spent
pre-multiplying vectors by $A$ and $A^T$. For \textit{flop-type} Poker
games like Texas Hold'em and  Rhode Island Hold'em,
$A$ (and thus $A^T$ too)  is very big (up $10^{14}$ rows and columns!)
but is sparse and has a rich block-diagonal structure (each block is
itself the Kronecker product of smaller matrices) which can be
carefully exploited, as in \cite{hoda2010smoothing}. Also the sampling
strategies presented in the recent work \cite{kroer2015} (section 6),
for generating unbiased estimates of $Ay$ and $A^Tx$ would readily
convert Algorithm \ref{Tab:algo} into an online and much scalable
solver.

\paragraph*{Computing $\|K\|$.}
A major ingredient in the proposed algorithm is $\|K\|$, the $2$-norm
of the huge matrix $K$. This  can be efficiently computed using the
power iteration. Also since $\|K\|$ is only used in defining the
step-size $\lambda := 1/\|K\|$, it may be possible to avoid computing
$\|K\|$ altogether, and instead use a line-search / backtrack strategy
(see \cite{o2013adaptive}, e.g) for setting $\lambda$.

\paragraph*{Game abstraction.} For many variants of Poker, there has
been extensive research in lossy / lossless abstraction techniques (for
example \cite{gilpin2007} and more recently,
\cite{sandholm2015abstraction,brown2015hierarchical}), wherein
strategically equivalent or not-so-different situations in the game
tree are lumped together. This can drastically reduce the size of the
state space from a player's perspective, and ultimately, the size of
the matrices $A$, $E_1$, and $E_2$, without significantly deviating
much from the true game.

\section{Numerical experiments results}
\label{sec:results}
We now present some proof-of-concept for the algorithm proposed.
Results are presented and commented in Figure \ref{Tab:dgap_curve}.
\begin{remark}
\label{thm:bulletproof}
We have not benchmarked our algorithm against the algorithms proposed in
\cite{nesterov2005a} and Gilpin's et al. \cite{gilpinfirst} because
implementing them from scratch for such games would require us to
compute the complicated projections $\prox_{Q_k}$. We recall that
avoiding these projections was one of the goals of the manuscript.
\end{remark}
\subsection{Basic test-bed: Matrix games on simplexes}
As in \cite{nesterov2005a,chambolle2014ergodic}, we generate a $1000
\times 1000$ random matrix whose entries are uniformly identically
distributed in the closed interval $[-1, 1]$. 
The results of the experiments are shown in Figure
\ref{Tab:dgap_curve}\textit{(a)}.

\subsection{Kuhn Poker, a ``toy'' sequential game} This game is a
simplified form of Poker
developed by Harold W. Kuhn in \cite{kuhn}. It already contains all
the complexities (sequentiality, imperfection of information, etc.) of
a full-blown Poker game like Texas Hold'em, but is simple enough to
serve as a proof-of-concept for the ideas developed in this
manuscript. The deck includes only three playing
cards: a King, Queen, and Jack. One card is dealt to each player, then
the first player must bet or pass, then the second player may bet or
pass. If any player chooses
to bet the opposing player must bet as well ("call") in order to stay
in the round. After both players pass or bet, the player with the
highest card wins the pot. 
The pair of vectors $(x^*, y^*) \in \mathbb{R}^{13 + 13}$ given by
\begin{eqnarray*}
  \begin{split}
    &x^* = [1, .759, .759, 0, .241, 1, .425, .575, 0, .275, 0,
      .275, .725]^T,\\
    &y^* = [1, 1, 0, .667, .333, .667, .333, 1, 0, 0, 1, 0, 1]^T
    \end{split}
\end{eqnarray*}
is a Nash $(10^{-4},0)$-equilibrium computed in 1500 iterations of
Algorithm  \ref{Tab:algo}. The convergence curves are shown
in Fig \ref{Tab:dgap_curve}. One easy checks that this equilibrium is
feasible. Indeed, one computes \\\\
$E_1x^* - e_1 = [4.76 \times 10^{-5}, -1.91 \times 10^{-5}, 5.67
      \times 10^{-5}, 8.23 \times 10^{-6}, 2.90 \times 10^{-5},
      -8.62 \times 10^{-7}, -1.96 \times 10^{-5}]^T$
and
$E_2y^* - e_2 = [-7.04 \times 10^{-7}, 2.27 \times 10^{-6}, -3.29
  \times 10^{-6}, -1.50 \times 10^{-6},
      2.92 \times 10^{-6}, -4.97 \times 10^{-7}, -5.85 \times
      10^{-7}]^T$.\\\\
Finally, one checks that ${x^*}^TAy^* = {-0.05555}$,
 which agrees to 5 d.p with the value of $-1 / 18$ computed
 analytically by H. W. Kuhn in his 1950 paper \cite{kuhn}. The
 evolution of the dual gap and the expected value of
 the game across iterations are shown in Figure \ref{Tab:dgap_curve}.
The results of the experiments are shown in Figure
\ref{Tab:dgap_curve}\textit{(b)}.

\begin{figure}[!htpb]
  \subfigure[$10^3 \times 10^3$ matrix game on simplexes]{
    \includegraphics[width=.5\linewidth]{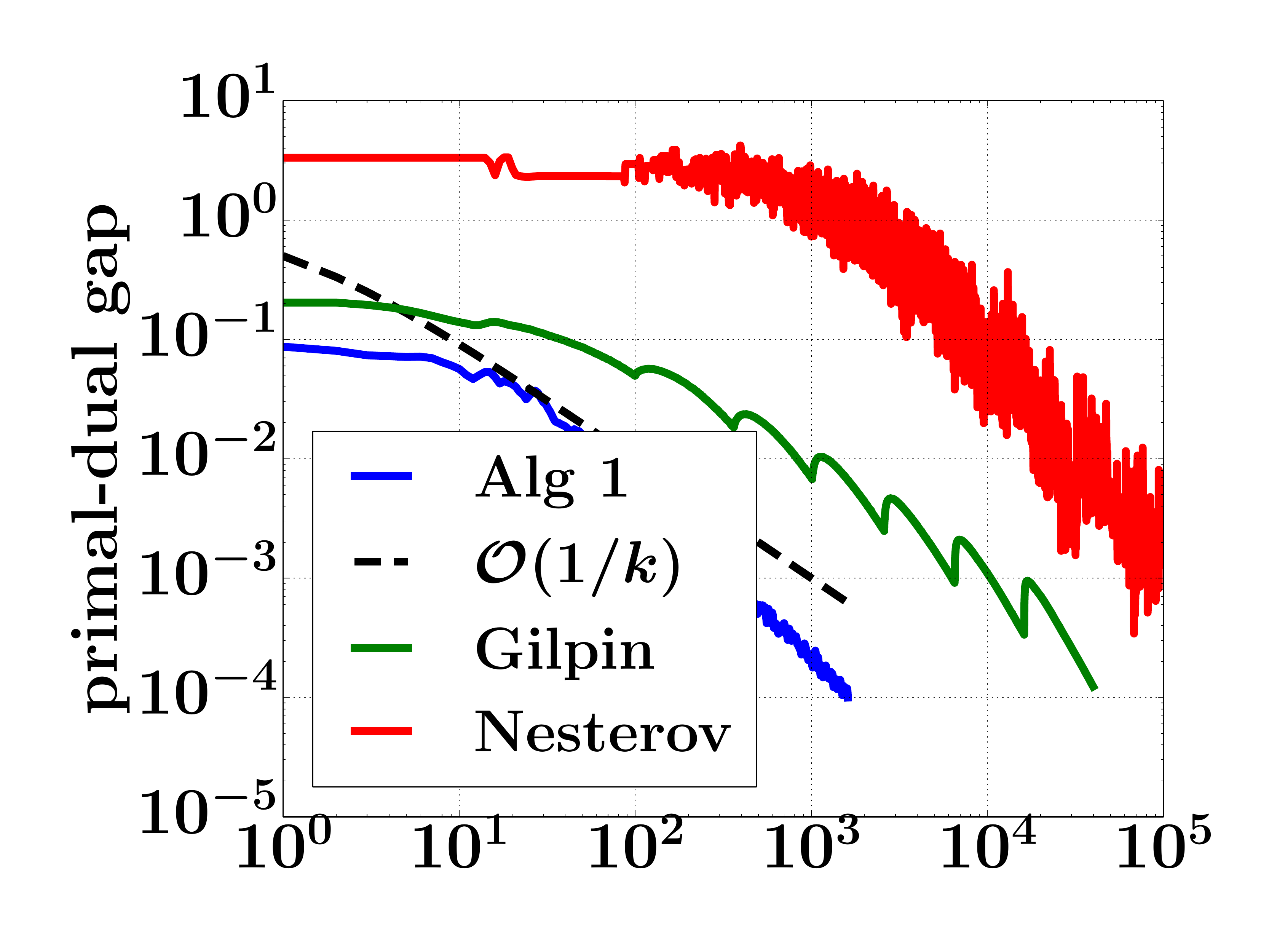}
  }
  \hspace{-2em}
  \subfigure[Kuhn 3-card Poker]{
    \includegraphics[width=.5\linewidth]{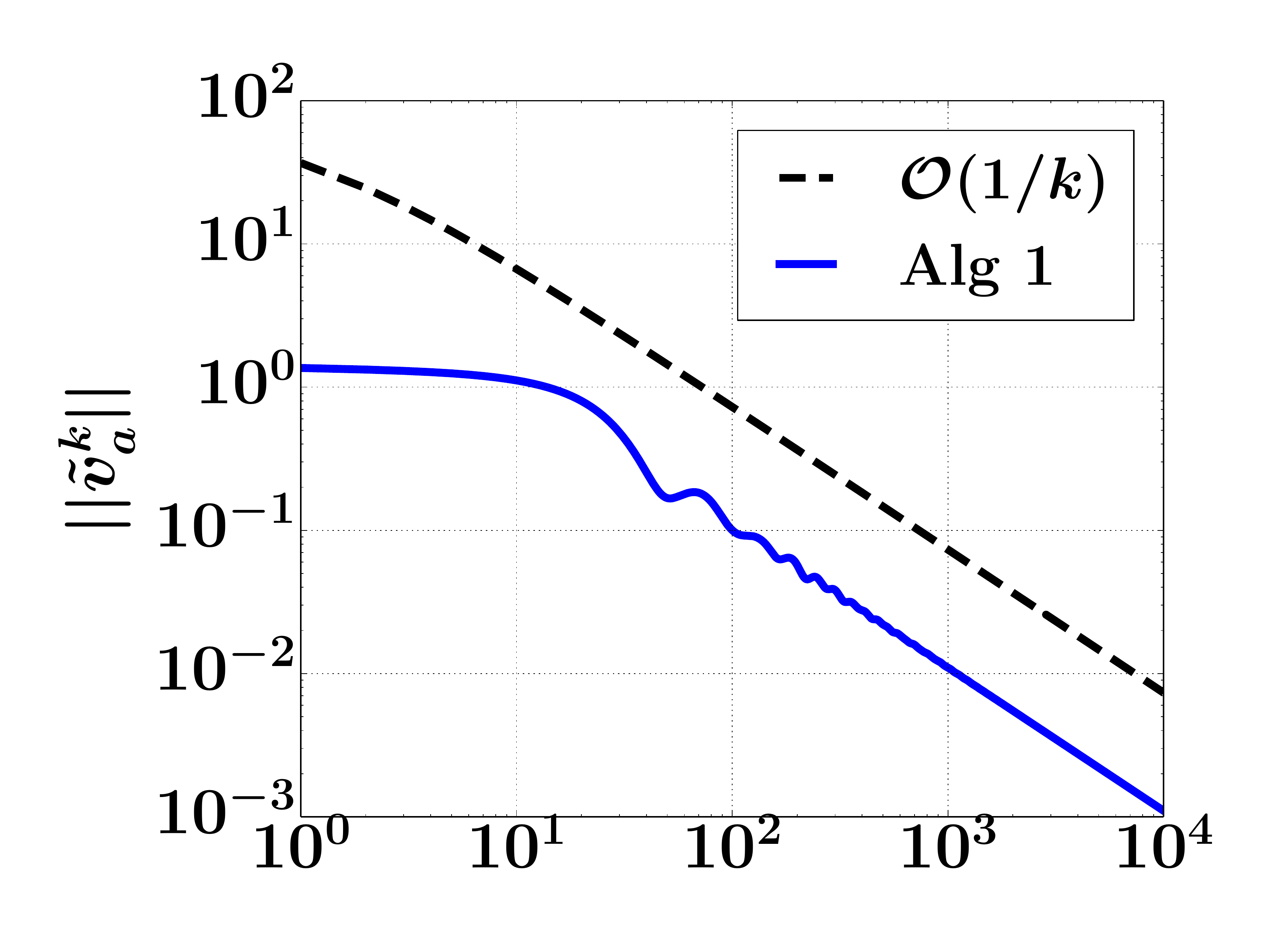}
  }
  \hspace{-.5em}
  \includegraphics[width=.5\linewidth]{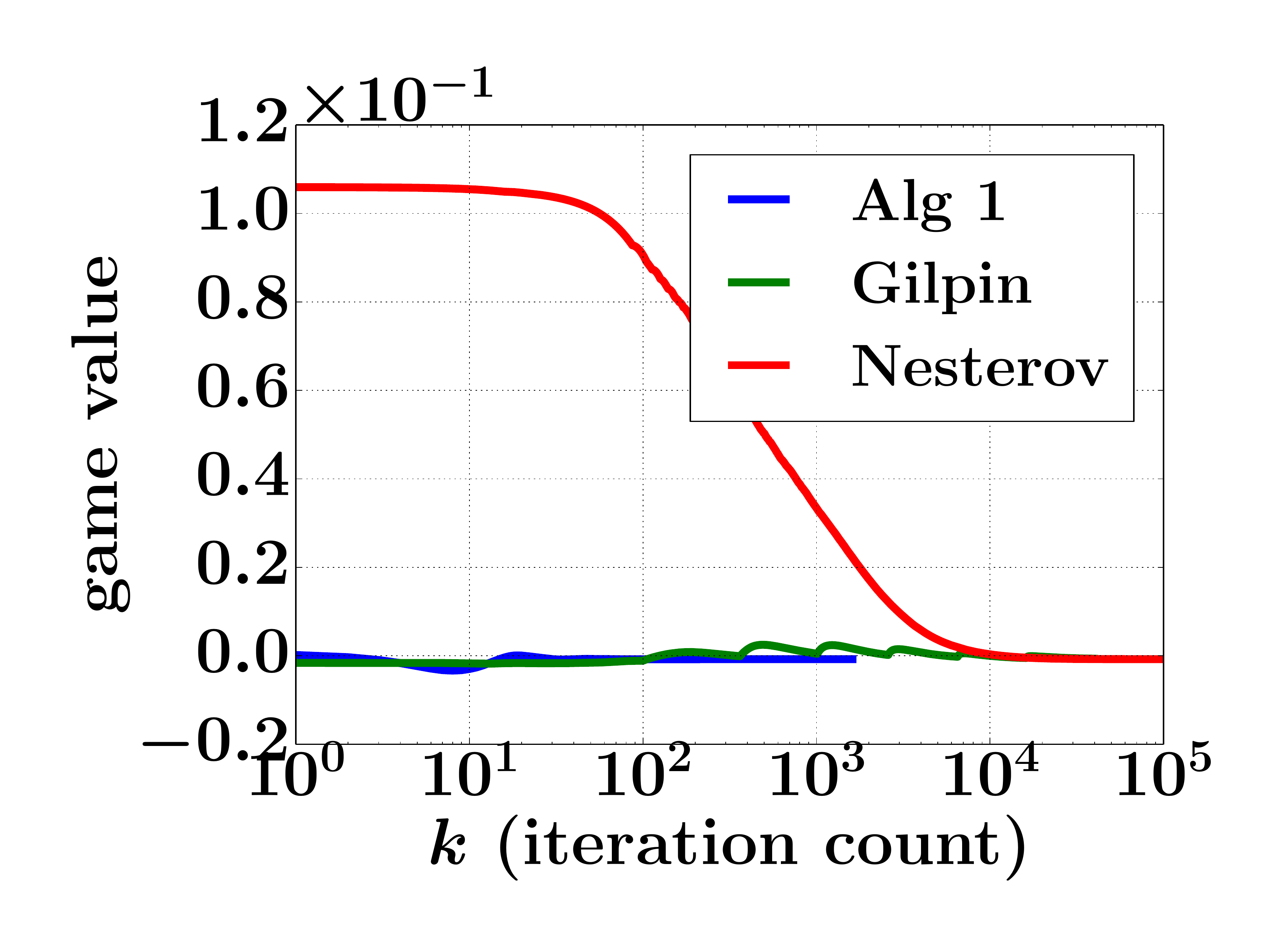}
  \includegraphics[width=.49\linewidth]{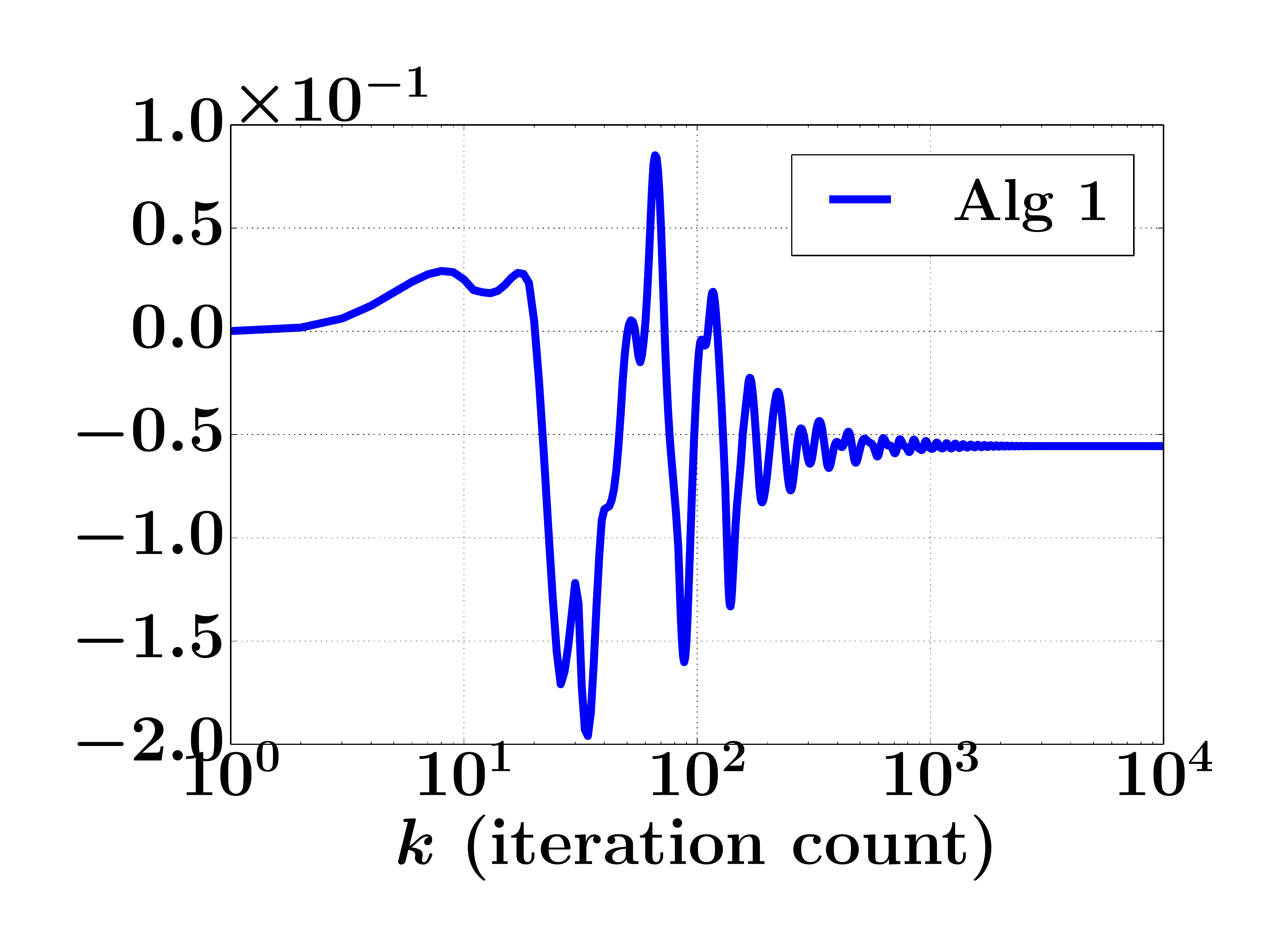}
  \caption{Convergence curves of Algorithm
    \ref{Tab:algo}.  We stress that the algorithms of Nesterov
    \cite{nesterov2005a} and Gilpin \cite{gilpinfirst} are included in
    the plots only indicatively, since this is not meant to be a
    benchmark as already explained (Remark \ref{thm:bulletproof}). In
    \textit{(a)}, the duality gaps
    are computed according to formula \eqref{eq:mg_pd}. One can see
    the linear (i.e exponentially fast) behavior of the algorithm in
    \cite{gilpinfirst}, inbetween consecutive breakpoints on the
    $\epsilon$ grid (though the rate of linear convergence seems
    to by quite close to $1$ here). As expected, the first-order smoothing
    algorithm labelled ``Nesterov'' \cite{nesterov2005a} jitters
    around as the iterations go on because even the smoothed problem
    becomes heavily ill-conditioned
    near solutions. 
  \textit{(b)}: Kuhn Poker. In the top-right plot, we show the
  modified duality gap defined in \eqref{eq:v_a} In both cases, we see
  that the proved convergence rate for our algorithm is empirically
  observed.}
  \label{Tab:dgap_curve}
\end{figure}

\section{Concluding remarks and future work}
Making use of the sequence-form representation
\cite{koller1992complexity,von1996efficient,vonequilibrium}, we have
devised a simple numerically stable primal-dual algorithm for computing
Nash-equilibria in two-person zero-sum sequential games with
incomplete information (like Texas Hold'em, etc.). Our algorithm is
simple to implement, with a low constant cost per iteration, and
enjoys a rigorous convergence theory with a proven
$\mathcal{O}(1/\epsilon)$ convergence in terms of basic operations
(matvec products, clipping, etc.), to a Nash
$(\epsilon,0)$-equilibrium of the game. In future, we plan to run more
experiments on real Poker games to measure the practical power of the
proposed algorithm compared to other competed schemes like CFR and
EGT.

In conclusion, Nash-equilibrium problems are saddle-point
convex-concave problems, and as such, a natural tool
for tackling them would be proximal primal-dual / operator-splitting
algorithms, and we believe such methods will receive more attention
in the algorithmic game theory community in future.







\pagebreak
\bibliographystyle{IEEEbib}
\bibliography{bib}

\end{document}